\documentclass[journal,twocolumn]{IEEEtran}

\usepackage[utf8]{inputenc} 
\usepackage[T1]{fontenc}
\usepackage{url, ifthen, cite}
\usepackage[cmex10]{amsmath}
\usepackage{graphicx, graphics, balance, dsfont}
\usepackage{amssymb, amsfonts, amsthm}
\usepackage{epsfig, tikz, standalone}

\newcommand{\PP}{\mathbb{P}}
\newcommand{\E}{\mathbb{E}}
\newcommand{\al}{\alpha}

\newtheorem{theorem}{Theorem}
\newtheorem{corollary}{Corollary}
\newtheorem{lemma}{Lemma}

\newtheorem{proposition}{Proposition}

\allowdisplaybreaks

\begin{document}
\title{SWIPT with Intelligent Reflecting Surfaces\\ under Spatial Correlation}

\author{Constantinos Psomas, \IEEEmembership{Senior Member, IEEE}, and Ioannis Krikidis, \IEEEmembership{Fellow, IEEE}
        
\thanks{C. Psomas and I. Krikidis are with the Department of Electrical and Computer Engineering, University of Cyprus, Nicosia, Cyprus; e-mail: \{psomas, krikidis\}@ucy.ac.cy. This work was co-funded by the European Regional
Development Fund and the Republic of Cyprus through the Research and
Innovation Foundation, under the projects INFRASTRUCTURES/1216/0017 (IRIDA) and POST-DOC/0916/0256 (IMPULSE). This work has also received
funding from the European Research Council (ERC) under the European
Union's Horizon 2020 research and innovation programme (Grant agreement
No. 819819).}}

\maketitle

\begin{abstract}
Intelligent reflecting surfaces (IRSs) can be beneficial to both information and energy transfer, due to the gains achieved by their multiple elements. In this work, we deal with the impact of spatial correlation between the IRS elements, in the context of simultaneous wireless information and power transfer. The performance is evaluated in terms of the average harvested energy and the outage probability for random and equal phase shifts. Closed-form analytical expressions for both metrics under spatial correlation are derived. Moreover, the optimal case is considered when the elements are uncorrelated and fully correlated. In the uncorrelated case, random and equal phase shifts provide the same performance. However, the performance of correlated elements attains significant gains when there are equal phase shifts. Finally, we show that correlation is always beneficial to energy transfer, whereas it is a degrading factor for information transfer under random and optimal configurations.
\end{abstract}


\section{Introduction}
\IEEEPARstart{I}{ntelligent} reflecting surfaces (IRSs) is a technology that has recently received notable consideration by the research community. An IRS refers to a planar array
of flat - and mostly passive - elements, which reflect the incident signal with the help of a dedicated controller \cite{DAI}. The reconfigurable nature of this technology has the potential of improving both coverage and energy efficiency. Moreover, the fact that an IRS can be embedded on a flat surface, makes it an attractive solution for both indoor and outdoor applications.

The benefits of IRSs have been exploited in numerous communication scenarios, e.g. \cite{QN,RUI,CP}, including simultaneous wireless information and power transfer (SWIPT) \cite{PAN,SHI,WU}. In SWIPT systems, the transmitted signal is used to convey both information and energy. Specifically, information decoding and energy harvesting is achieved at the receiver with the employment of a practical scheme such as time-switching or power-splitting (PS) \cite{ZENG}. Thus, the implementation of SWIPT over an IRS, can provide significant gains to both information decoding and energy harvesting. The work in \cite{PAN}, considers a SWIPT system, where an IRS assists in the communication towards information receivers as well as energy transfer to energy receivers. In \cite{SHI}, the authors studied IRS-aided SWIPT systems in the context of secrecy, where the transmit beamforming and phases shifts are jointly optimized to maximize the harvested energy. Moreover, the weighted sum-power maximization problem in a SWIPT system with an IRS was considered in \cite{WU}; the proposed design is shown to enhance the rate-energy performance trade-off.

Since an IRS consists of multiple adjacent elements, it implies that the elements exhibit a certain spatial correlation between them \cite{QN, EB}. However, the effect of correlation on the outage probability and average harvested energy has not been previously studied. Therefore, in this work, we focus on the impact of channel correlation between the elements of an IRS and study the performance from a SWIPT point-of-view. In particular, we evaluate the performance of information and power transfer under spatial correlation for two configurations: random and equal phase shifts. Closed-form expressions for the average harvested energy are analytically derived, as well as a closed-form approximation for the outage probability. Exact expressions for the outage probability are provided for the two extreme cases, namely, when the elements are mutually uncorrelated and fully correlated. The optimal phase configuration is also considered for the two extreme cases. Our results show that correlation is always beneficial to energy harvesting. On the other hand, we show that correlation is a degrading factor to the outage probability with a random or optimal configuration.

\section{System Model}\label{sys_model}
\subsubsection{Topology and Channel Model}
We consider a simple topology, consisting of a transmitter, a receiver and an IRS located between them. The transmitter and the receiver have a single antenna and the IRS is equipped with $M$ reflecting elements. It is assumed that a direct link between the transmitter and the receiver is not available (e.g., due to heavy shadowing) \cite{QN}. All wireless links are assumed to exhibit Rayleigh fading. We define by $h_i$ the channel coefficient between the transmitter and the $i$-th IRS element, and by $g_i$ the channel coefficient between the $i$-th element and the receiver; these are assumed to be distributed according to a complex Gaussian distribution with zero mean and unit variance, i.e. $h_i \sim \mathcal{CN}(0,1)$ and $g_i \sim \mathcal{CN}(0,1)$.

Time is slotted and, at each time slot, the $i$-th element shifts the phase of the incident signal to a certain angle $\phi_i \in [0,2\pi)$. We assume that any two elements are spatially correlated. In particular, we denote by $\rho_{i,k}$ the correlation coefficient for the $i$-th and $k$-th elements, $i, k \in \{1,\dots,M\}$, with $0 \leq \rho_{i,k} \leq 1$, $i\neq j$, and $\rho_{i,i} = 1, \forall i$. Let $\mathbf{R}$ be the correlation matrix, whose $(i,k)$-th element is equal to $\rho_{i,k}$. Then, the received signal at the receiver is written as
$r = \sqrt{P_t (d_1 d_2)^{-\al}} \mathbf{R}^\frac{1}{2} \mathbf{h}^T \mathbf{\Phi} \mathbf{g} \mathbf{R}^\frac{1}{2} x + n$, where $P_t$ is the transmit power, $d_1$ ($d_2$) is the transmitter to IRS (IRS to receiver) distance, $\al$ is the path-loss exponent, $x$ is the transmitted symbol, $\mathbf{h} = [h_1 ~ h_2 ~ \dots ~ h_M]^T$, $\mathbf{g} = [g_1 ~ g_2 ~ \dots ~ g_M]^T$, $n \sim \mathcal{CN}(0,\sigma^2)$ is an additive white Gaussian noise (AWGN) with variance $\sigma^2$ and $\mathbf{\Phi} = \text{diag}[e^{\jmath\phi_1} ~ e^{\jmath\phi_2} ~ \dots ~ e^{\jmath\phi_M}]$ is the diagonal matrix having the phase shift variables. Hence,
\begin{align}\label{gain}
H = \Big\lvert \textstyle\sum_{i=1}^M h_i g_i e^{\jmath\phi_i}\Big\rvert^2,
\end{align}
provides the end-to-end channel gain achieved by the $M$ elements of the IRS.

\subsubsection{Information and Energy Transfer}
The receiver has SWIPT capabilities and splits the power of the received signal into two parts; one part is converted to baseband for information decoding and the other is directed to the rectenna for energy harvesting \cite{ZENG}. Let $0 < \zeta \leq 1$ denote the power-splitting parameter, i.e. $100\zeta\%$ of the received power is used for decoding. The additional circuit noise formed during the baseband conversion phase is modeled as an AWGN with zero mean and variance $\sigma^2_c$. Therefore, the achieved instantaneous signal-to-noise ratio (SNR) at the receiver is
\begin{align}\label{snr}
\eta = \frac{\zeta P_t (d_1 d_2)^{-\al} H}{\zeta\sigma^2 + \sigma^2_c},
\end{align}
where $H$ is given by \eqref{gain}.

Moreover, as $100(1 - \zeta)\%$ of the received signal is passed to the rectifier, the harvested energy is modeled by the following quadratic polynomial \cite{ZENG}
\begin{align}\label{eh}
E = (1-\zeta)P_t (d_1 d_2)^{-\al} H + (1-\zeta)^2P_t^2 (d_1 d_2)^{-2\al} H^2.
\end{align}
Note that any energy harvesting from the AWGN noise is considered to be negligible.

\section{IRS-aided SWIPT under Spatial Correlation}\label{analysis}
In this section, we analytically evaluate the impact of spatial correlation on IRS-aided SWIPT. The analysis is undertaken for two phase shift configurations, namely, random phase shifts and equal phase shifts. Nevertheless, the methodology can be easily adapted to consider other phase configurations or channel models. We first focus on the energy harvesting aspect.

\subsection{Average Harvested Energy}
It follows from \eqref{eh}, that the average harvested energy is
\begin{align}\label{aeh}
\mathcal{E} = (1-\zeta)&P_t (d_1 d_2)^{-\al}(\E_H\{H\}\nonumber\\
&+ (1-\zeta)P_t (d_1 d_2)^{-\al} \E_H\{H^2\}),
\end{align}
where $E_H\{H\}$ and $E_H\{H^2\}$ are the first and second (raw) moments of $H$, respectively. We first state the following lemma for the random phase shifts.

\begin{lemma}\label{moments_random}
If the IRS employs random phase shifts\footnote{In practical IRS implementations, the phase shifts take a finite number of discrete values \cite{WU2}. However, the performance of random phase shifts is not affected by the number of available values \cite{CP}.}, the first and second moments of $H$ are given by
\begin{align}
\mu_1 = M,
\end{align}
and
\begin{align}\label{mu2_random}
\mu_2 = 4M + 2\sum_{i\neq k} (\rho_{i,k}^2+1)^2,
\end{align}
respectively.
\end{lemma}

\begin{proof}
	See Appendix \ref{moments_random_prf}.
\end{proof}

A direct consequence of the above lemma is given in the following corollary.

\begin{corollary}
For a uniform correlation, i.e. $\rho_{i,k} = \rho$, $i,k \in \{1,\dots,M\}$, $i\neq k$, Eq. \eqref{mu2_random} is reduced to
\begin{align}
\mu_2 = 2M((M-1)(\rho^2+1)^2 + 2).
\end{align}
\end{corollary}

Obviously, for $\rho = 0$ (uncorrelated) and $\rho = 1$ (fully correlated), we have $\mu_2 = 2M(M+1)$ and $\mu_2 = 4M(2M-1)$, respectively. Furthermore, it is clear to see that at large values of $M$, the second moment scales with $M^2$.

\begin{lemma}\label{moments_equal}
If the IRS employs equal phase shifts, i.e. $\phi_i = \phi_k$, $\forall i,k \in \{1,\dots,M\}$, the first moment of $H$ is
\begin{align}
\mu_1 = M + \sum_{i, k}\rho_{i,k}^2,\label{mu1_equal}
\end{align}
and the second moment is given by
\begin{align}
\mu_2 = 4M &+ 2\sum_{i, k} (\rho_{i,k}^2+1)^2 + 2\rho_{i,k}^2(\rho_{i,k}^2+4)\nonumber\\
&+4\sum_{i, k, l} (\rho_{i,k}\rho_{i,l} + \rho_{k,l})^2 + 2 \rho_{i,k}^2\rho_{i,l}^2\nonumber\\
&+\sum_{i, k, l, m} (\rho_{i,k}\rho_{l,m} + \rho_{i,m}\rho_{k,l})^2,\label{mu2_equal}
\end{align}
where $i,k,l,m \in \{1,\dots,M\}$ are mutually unequal.
\end{lemma}

\begin{proof}
See Appendix \ref{moments_equal_prf}.
\end{proof}

In contrast to the random case, the mean in this scenario is positively affected by the correlation.

\begin{corollary}
For a uniform correlation, i.e. $\rho_{i,k} = \rho$, $i,k \in \{1,\dots,M\}$, $i\neq k$, Eq. \eqref{mu1_equal} reduces to
\begin{align}
\mu_1 = M((M-1)\rho^2+1),
\end{align}
and Eq. \eqref{mu2_equal} is reduced to
\begin{align}
\mu_2 &= 4M + 4\binom{M}{2}(3\rho^4+10\rho^2+1)\nonumber\\
&\quad+ 24\binom{M}{3}(3\rho^4+2\rho^3+\rho^2) + 96\binom{M}{4}\rho^4.
\end{align}
\end{corollary}

As such, $\rho = 0$ results in $\mu_1 = M$, $\mu_2 = 2M(M+1)$ and $\rho = 1$ gives $\mu_1 = M^2$, $\mu_2 = 4M^4$. Observe that for $\rho = 1$ (full correlation), $\mu_2$ increases proportionally to $M^4$, which signifies the massive gains that can be achieved in energy harvesting. Based on \eqref{aeh}, we can now state the following.

\begin{theorem}
The achieved average harvested energy is
\begin{align}\label{aheu}
\mathcal{E} = (1-\zeta) P_t (d_1 d_2)^{-\al} (\mu_1 + (1-\zeta)P_t (d_1 d_2)^{-\al} \mu_2),
\end{align}
where $\mu_1$ and $\mu_2$ are either given by Lemma \ref{moments_random} or Lemma \ref{moments_equal}.
\end{theorem}

From Lemma \ref{moments_random} and Lemma \ref{moments_equal}, we can see that both phase configurations achieve the same average harvested energy in the uncorrelated case ($\rho_{i,k} = 0, i \neq k$). On the other hand, for any $\rho_{i,k} > 0, i \neq k$, an equal phase shift configuration provides higher energy harvesting than the random configuration. In fact, for uniform correlation $\rho = 1$, the equal phase configuration achieves the same performance as the optimal phase configuration since
\begin{align}\label{opt1}
\E\left\{\left(\sum_{i=1}^M |h_i||g_i| \right)^2\right\} = \E\left\{M^2 |h_i|^2|g_i|^2\right\} = M^2,
\end{align}
which is equal to $\mu_1$ and
\begin{align}\label{opt2}
\E\left\{\left(\sum_{i=1}^M |h_i||g_i| \right)^4\right\} = \E\left\{M^4 |h_i|^4|g_i|^4\right\} = 4 M^4,
\end{align}
which corresponds to $\mu_2$. Keep in mind that the equal phase shift value depends on the considered channel model and should be adjusted accordingly in order to achieve the optimal performance, e.g. in the case of a uniform linear array, this value is a function of the angles of arrival and departure. For the sake of completeness, we also evaluate the optimal phase configuration with no correlation. Specifically, we have $\E\{(\sum_{i=1}^M |h_i||g_i|)^2\} = M + M(M-1)\pi^2/16$ \cite{CP} and, by using similar arguments as in Appendices \ref{moments_random_prf} and \ref{moments_equal_prf}, we obtain $\E\{(\sum_{i=1}^M |h_i||g_i|)^4\} = 4M + 6\binom{M}{2}(3\pi^2/16+1) + 9\binom{M}{3}\pi^2/4 + 3\binom{M}{4}\pi^4/32$.

\subsection{Outage Probability}
We now turn our attention to the information transfer. Let $\tau$ be a non-negative pre-defined rate threshold. Then,
\begin{align}\label{op}
P_o = \PP\left\{\log_2\left(1+\eta\right) < \tau\right\},
\end{align}
defines the information outage probability and where $\eta$ is the SNR, given by \eqref{snr}. The following theorem provides an approximation of the outage probability.

\begin{theorem}\label{pop}
The outage probability is approximated by $P_o \approx \frac{\gamma(\kappa,\xi/\theta)}{\Gamma(\kappa)}$, with $\kappa = \frac{\mu_1^2}{\mu_2 - \mu_1^2}$, $\theta = \frac{\mu_2 - \mu_1^2}{\mu_1}$, and
\begin{align}\label{xi}
\xi = \frac{(2^\tau - 1)(\zeta\sigma^2+\sigma^2_c)}{\zeta P_t (d_1 d_2)^{-\al}},
\end{align}
where $\mu_1$ and $\mu_2$ are given either by Lemma \ref{moments_random} or Lemma \ref{moments_equal} and $\Gamma(\cdot)$ and $\gamma(\cdot,\cdot)$ are the complete and lower incomplete gamma functions, respectively.
\end{theorem}

\begin{proof}
See Appendix \ref{op_prf}.
\end{proof}

As the moments of $H$ are larger with the equal phase configuration, it follows that it achieves a lower outage probability (smaller $\kappa$, larger $\theta$). Theorem \ref{pop} provides a closed-form expression for the outage probability and is a good approximation, as we show in the numerical results. Nevertheless, in what follows, we also provide exact analytical expression for the uncorrelated ($\rho_{i,k} = 0$, $i\neq k$) and fully correlated ($\rho_{i,k} = 1$) cases. The proposition below is given without a proof, as it follows directly from \cite{CP}.

\begin{proposition}\label{opu}
In the uncorrelated case, i.e. $\rho_{i,k} = 0$, $i\neq k$, the outage probability is given by
\begin{align}
P_o = 1-\frac{2}{\Gamma(M)} \xi^{\frac{M}{2}} K_M\left(2 \sqrt{\xi}\right),
\end{align}
where $\xi$ is given by \eqref{xi} and $K_M(\cdot)$ is the modified Bessel function of the second kind of order $M$.
\end{proposition}

The above expression is valid for both configurations (random or equal). However, for the correlated case, we need to consider them separately. In the case of random phase shifts, the channel gain is $H = |h|^2 |g|^2 \big\lvert \sum_{i=1}^M e^{\jmath\phi_i}\big\rvert^2$ and we can state the following.

\begin{proposition}\label{opc}
In the fully correlated case, i.e. $\rho_{i,k} = 1$, $\forall i,k$, with random phase shifts, the outage probability is given by
\begin{align}	P_o = 1-\frac{2}{M} \int_0^\infty \sqrt{\frac{\xi}{\psi}} e^{-\frac{\psi}{M}} K_1\left(2 \sqrt{\frac{\xi}{\psi}}\right) d\psi,
\end{align}
where $\xi$ is given by \eqref{xi}.
\end{proposition}

\begin{proof}
	See Appendix \ref{opc_prf}.
\end{proof}

Finally, for the correlated case with equal phase shifts, we have $H = M^2 |h|^2 |g|^2$. The proof for this case is omitted as it follows similar arguments with the previous results.

\begin{proposition}\label{opc2}
In the fully correlated case, i.e. $\rho_{i,k} = 1$, $\forall i,k$, with equal phase shifts, the outage probability is
\begin{align}
P_o = 1-\frac{2}{M} \sqrt{\xi} K_1\left(2 \frac{\sqrt{\xi}}{M}\right),
\end{align}
where $\xi$ is given by \eqref{xi}.
\end{proposition}

The above corresponds to the performance of the optimal phase configuration with full correlation. For the uncorrelated case, we can use the result from \cite[Theorem 3]{CP}.

\begin{figure}[t]\centering
	\includegraphics[width=0.9\linewidth]{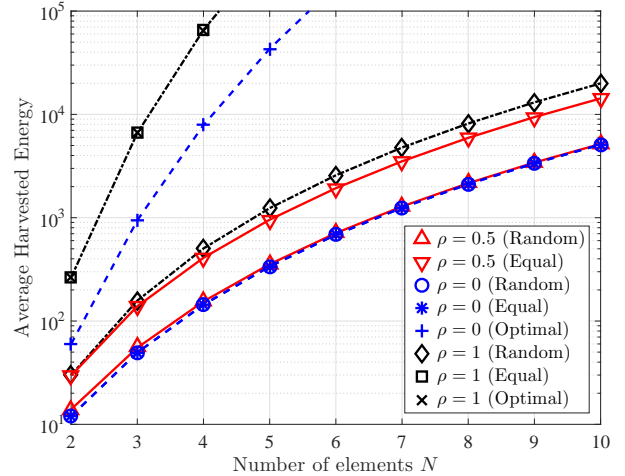}
	\caption{Average harvested energy in terms of $N$; simulation and theoretical results are depicted with markers and lines, respectively.}\label{fig1}
\end{figure}

\section{Numerical Results \& Discussion}\label{numerical}
We verify the above analytical framework with Monte Carlo simulations. The following parameters are used: $P_t = 0$ dB, $\sigma^2 = 0$ dB, $\sigma^2_c = 0$ dB, $d_1=d_2=1$ m, $\al=4$, $\zeta = 0.5$ and $\tau = 1$ bpcu. Moreover, we consider a square IRS with $M = N\times N$ elements. Finally, we assume $\rho_{i,k} = \rho^{|i-k|}$, $0 \leq \rho \leq 1$, where $|i-k|$ describes the distance between elements $i$ and $k$ \cite{QN}. In other words, the correlation between two elements decreases as the distance between them increases. The maximum correlation between distinct elements is equal to $\rho$ and occurs between adjacent elements. Note that other spatial correlation models can be considered, which would provide similar observations.

Fig. \ref{fig1} illustrates the average harvested energy for the considered configurations, in terms of $N$. In all cases, the uncorrelated and fully correlated scenarios correspond to the lower and upper bound, respectively. For $\rho=0$, the random and equal configurations achieve the same performance. On the other hand, when $\rho > 0$, the equal phase configuration outperforms the random case. Moreover, with random phases and $\rho=0.5$, the performance quickly converges to the uncorrelated case as $M$ increases. This is because the correlation between the elements is reduced, i.e., $\rho^{|i-k|} \to 0$ as $|i-k|\to \infty$ for $\rho < 1$. Finally, the simulation results (depicted with markers) match our theoretical expressions (depicted with lines), which verifies our analytical framework.

Fig. \ref{fig2} depicts the outage probability achieved by the considered configurations. Firstly, for $\rho=0.5$, the figure shows that Theorem \ref{pop} is a good approximation of the outage probability. Similarly to the average harvested energy, the equal case outperforms the random one. Furthermore, under random and optimal phase configurations, the lowest outage probability is achieved when the elements are uncorrelated. In other words, contrary to energy harvesting (which is a long term operation), correlation is a degrading factor for information transfer with these configurations. The opposite occurs when equal phase shifts are considered, that is, the fully correlated case achieves the best performance. 

The above observations indicate that the impact of spatial correlation differs, depending on the phase configuration and the considered performance metric. This remark is general and is not restricted to the configurations in this work. Hence, an IRS could be engineered in such a way, e.g. the number of elements, the elements' size and dimensions as well as the implemented topology, so as to satisfy a certain correlation requirement. Also, more advanced IRSs could employ adaptive techniques, where elements can dynamically be switched on/off \cite{CP} or can be reconfigured through a software-controllable fluidic structure \cite{KIT}, in order to achieve the required correlation.

\section{Conclusion}\label{conclusion}
In this work, we focused on the effect of spatial correlation between the elements of an IRS in the context of SWIPT. We evaluated the performance of power transfer and the performance of information transfer under spatial correlation with random, equal and optimal phase configurations. Closed-form expressions for the average harvested energy were analytically derived together with a closed-form approximation for the outage probability. We showed that correlation always benefits energy harvesting but information transfer can be negatively affected in the presence of correlation.

\begin{figure}[t]\centering
  \includegraphics[width=0.9\linewidth]{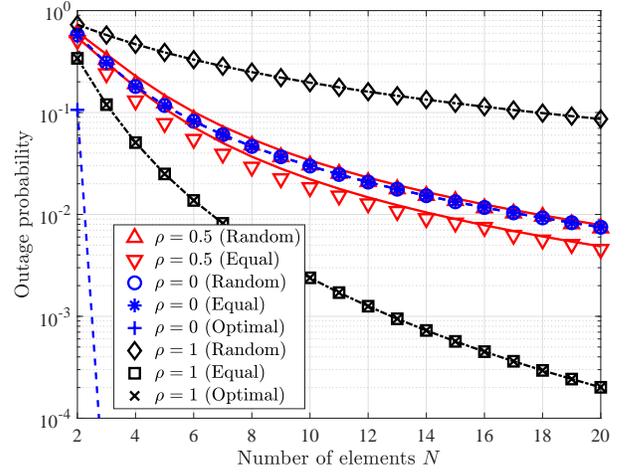}
  \caption{Outage probability with respect to $N$; simulation and theoretical results are depicted with markers and lines, respectively.}\label{fig2}
\end{figure}

\appendix
\subsection{Proof of Lemma \ref{moments_random}}\label{moments_random_prf}
We first provide the multivariate moment that is needed for the proofs of both Lemma \ref{moments_random} and Lemma \ref{moments_equal}. In particular, by using Isserlis' theorem \cite{PS}
\begin{align}
	&\E\{h_i h_k h_l^* h_m^*\} = \E\{h_i h_k\}\E\{h_l^* h_m^*\} + \E\{h_i h_l^*\}\E\{h_k h_m^*\}\nonumber\\
	&\qquad\qquad+ \E\{h_i h_m^*\}\E\{h_k h_l^*\} = \rho_{i,l} \rho_{k,m} + \rho_{i,m} \rho_{k,l},\label{multivar}
\end{align}
which follows from $\E\{h_i h_k\} = 0$ and $\E\{h_i h_l^*\} = \rho_{i,l}$, $\forall i,k,l$.

The first moment $\mu_1$ of $H$ is written as
\begin{align}
\mu_1 &= \E\{H\} = \E\left\{\left(\sum_{i=1}^M h_i g_i e^{\jmath\phi_i}\right)\left(\sum_{i=1}^M h_i^* g_i^* e^{-\jmath\phi_i}\right)\right\}\nonumber\\
&=\E\left\{\sum_{i=1}^M |h_i|^2 |g_i|^2 + \sum_{i\neq k} h_i g_i h_k^* g_k^* e^{\jmath(\phi_i-\phi_k)}\right\}.\label{eq1}
\end{align}
The first term is evaluated as follows
\begin{align}
\E\left\{\sum_{i=1}^M |h_i|^2 |g_i|^2\right\} = \sum_{i=1}^M \E\left\{|h_i|^2 |g_i|^2\right\} = M,
\end{align}
since $\E\left\{|h_i|^2\right\} = \E\left\{|g_i|^2\right\} = 1$. Due to the random phase shifts, the second term is zero since $\E\{e^{\jmath(\phi_i-\phi_k)}\} = 0$.

From \eqref{eq1}, we can derive the second moment as
\begin{align}
&\mu_2 = \E\{H^2\}\nonumber\\
&= \E\Bigg\{\Bigg(\sum_{i=1}^M |h_i|^2 |g_i|^2 + \sum_{i\neq k} h_i g_i h_k^* g_k^* e^{\jmath(\phi_i - \phi_k)}\Bigg)^2\Bigg\}.\label{eq2}
\end{align}
By expanding the above polynomial, we look at each term separately. Firstly, we have
\begin{align}
&\E\left\{\left(\sum_{i=1}^M |h_i|^2 |g_i|^2\right)^2\right\} = \E\left\{\sum_{i=1}^M |h_i|^4 |g_i|^4\right\}\nonumber\\
&\hspace{30mm}+ \E\left\{\sum_{i\neq k} |h_i|^2 |g_i|^2 |h_k|^2 |g_k|^2\right\}.\label{prf3}
\end{align}
Then,
\begin{align}
\E\left\{\sum_{i=1}^M |h_i|^4 |g_i|^4\right\} = \sum_{i=1}^M \E\left\{|h_i|^4 |g_i|^4\right\} = 4M,\label{prf1}
\end{align}
which follows from $\E\left\{|h_i|^4\right\} = \E\left\{|g_i|^4\right\} = 2$, and
\begin{align}
\E\left\{\sum_{i\neq k} |h_i|^2 |g_i|^2 |h_k|^2 |g_k|^2\right\} = \sum_{i\neq k} (\rho_{i,k}^2+1)^2,\label{prf2}
\end{align}
which follows from the fact that $\E\{|h_i|^2 |h_k|^2\} = \rho_{i,k}^2 + 1$ (see \eqref{multivar}). The second term of the expansion is
\begin{align}
&\!\!\E\left\{\left(\sum_{i\neq k} h_i g_i h_k^* g_k^* e^{\jmath(\phi_i - \phi_k)}\right)^2\right\}\nonumber\\
&\!\!= \E\left\{\sum_{i\neq k, l\neq m} h_i g_i h_k^* g_k^* h_l g_l h_m^* g_m^* e^{\jmath(\phi_i + \phi_l - \phi_k - \phi_m)}\right\},
\end{align}
which is non-zero only if $i = m$ and $k = l$. Thus, we have
\begin{align}
\E\left\{\sum_{i\neq k} |h_i|^2 |g_i|^2 |h_k|^2 |g_k|^2\right\} = \sum_{i\neq k} (\rho_{i,k}^2+1)^2,\label{prf4}
\end{align}
which follows as \eqref{prf2}. Finally, by using the same argument as before, the last term of the expansion is zero. Hence, the final result is obtained by substituting \eqref{prf1}, \eqref{prf2} and \eqref{prf4} in \eqref{prf3}.

\subsection{Proof of Lemma \ref{moments_equal}}\label{moments_equal_prf}
The proof follows similar steps as the one for random phase shifts. However, in this case, the second term in \eqref{eq1} is not necessarily zero. Therefore, for the derivation of both moments, we focus on the effect of this term. Specifically,
\begin{align}
&\E\left\{\sum_{i\neq k} h_i g_i h_k^* g_k^* e^{\jmath(\phi_i-\phi_k)}\right\} = \E\left\{\sum_{i\neq k} h_i g_i h_k^* g_k^*\right\}\nonumber\\
&= \sum_{i\neq k} \E\left\{h_i h_k^*\right\} \E\left\{g_i g_k^*\right\} = \sum_{i\neq k}\rho_{i,k}^2,
\end{align}
where $\E\{h_i h_k^*\} = \E\{g_i g_k^*\} = \rho_{i,k}$, since $h_i$ and $g_i$ are correlated to $h_k^*$ and $g_k^*$, respectively, which gives the first moment.

We now turn our attention to the second moment and evaluate the remaining two terms from the polynomial expansion of \eqref{eq2}. In particular, we first have
\begin{align}
&\E\left\{\!\left(\sum_{i\neq k} h_i g_i h_k^* g_k^*\right)^{\!2}\right\} = \E\left\{\sum_{i\neq k, l\neq m} \!\!h_i g_i h_k^* g_k^* h_l g_l h_m^* g_m^*\right\}\!,
\end{align}
where there are five cases to consider. The first one, $i = m$ and $k = l$, is given by \eqref{prf4}. Next, if $i = l$ and $k = m$,
\begin{align}
\E\left\{\sum_{i\neq k} h_i^2 g_i^2 h_k^{*2} g_k^{*2}\right\} = \sum_{i\neq k} (2\rho_{i,k}^2)^2,
\end{align}
since $\E\{h_i^2 h_k^{*2}\} = 2\rho_{i,k}^2$ (see \eqref{multivar}). If $i = l$ and $k \neq m$ (or $i \neq l$ and $k = m$),
\begin{align}
\E\left\{\sum_{i\neq k} h_i^2 g_i^2 h_k^* g_k^* h_l^* g_l^*\right\} = \sum_{i\neq k} (2\rho_{i,k}\rho_{i,l})^2,
\end{align}
as $\E\{h_i^2 h_k^* h_l^*\} = 2\rho_{i,k}\rho_{i,l}$. Moreover, if $i = m$ and $k \neq l$ (or $i \neq m$ and $k = l$), we have
\begin{align}
\E\left\{\sum_{i\neq k\neq l} |h_i|^2 |g_i|^2 h_k g_k h_l^* g_l^* \right\} = \sum_{i\neq k\neq l} (\rho_{i,k}\rho_{i,l} + \rho_{k,l})^2,
\end{align}
which follows from $\E\left\{|h_i|^2 h_k h_l^*\right\} = \rho_{i,k}\rho_{i,l} + \rho_{k,l}$. Finally, if $i \neq k \neq l \neq m$,
\begin{align}
&\E\left\{\sum_{i\neq k\neq l\neq m} h_i g_i h_k^* g_k^* h_l g_l h_m^* g_m^*\right\}\nonumber\\
&\qquad\qquad\qquad\qquad= \sum_{i\neq k\neq l\neq m} (\rho_{i,k}\rho_{l,m} + \rho_{i,m}\rho_{k,l})^2,
\end{align}
as $\E\left\{h_i h_k^* h_l h_m^*\right\} = \rho_{i,k}\rho_{l,m} + \rho_{i,m}\rho_{k,l}$.

The final term of the expansion is written as
\begin{align}
&2\E\Bigg\{\sum_{i=1}^M |h_i|^2 |g_i|^2\sum_{i\neq k} h_i g_i h_k^* g_k^*\Bigg\}\nonumber\\
&= 2\E\Bigg\{\sum_{i\neq k} |h_i|^2 |g_i|^2 h_i g_i h_k^* g_k^* + \sum_{i\neq k} |h_i|^2 |g_i|^2 h_i^* g_i^* h_k g_k\nonumber\\
&\hspace{4cm}+ \sum_{i\neq k\neq l} |h_i|^2 |g_i|^2 h_k g_k h_l^* g_l^*\Bigg\}\nonumber\\
&= 2\left(2\sum_{i\neq k} (2\rho_{i,k})^2 + \sum_{i\neq k\neq l} (\rho_{i,k}\rho_{i,l} + \rho_{k,l})^2\right),
\end{align}
which follows from $\E\{|h_i|^2 h_i h_k^*\} = \E\{|h_i|^2 h_i^* h_k\} = 2\rho_{i,k}$ and $\E\{|h_i|^2 h_k h_l^*\} = \rho_{i,k}\rho_{i,l} + \rho_{k,l}$. The proof is completed by putting together all the above.

\subsection{Proof of Theorem \ref{pop}}\label{op_prf}
We will employ moment matching to approximate the distribution of $H$ with a gamma distribution. Specifically, a gamma distribution with mean $\mu$ and variance $\sigma^2$ has shape parameter $\kappa = \mu^2/\sigma^2$ and scale parameter $\theta = \sigma^2/\mu$. Therefore, by taking into account Lemma \ref{moments_random} or Lemma \ref{moments_equal}, we derive $\kappa = \frac{\mu_1^2}{\mu_2 - \mu_1^2}$ and $\theta = \frac{\mu_2 - \mu_1^2}{\mu_1}$. As such, $H$ can be approximated as a gamma random variable with shape and scale given above. From \eqref{op}, we have
$
P_o = \PP\left\{H < \frac{(2^\tau - 1)(\zeta\sigma^2+\sigma^2_c)}{\zeta P_t (d_1 d_2)^{-\al}}\right\},
$
and the result follows by using the cumulative distribution function of a Gamma random variable.

\subsection{Proof of Proposition \ref{opc}}\label{opc_prf}
Let $\psi = \big\lvert \sum_{i=1}^M e^{\jmath\phi_i}\big\rvert^2$. Then, by employing the central limit theorem, it is easy to deduce that $\psi$ is exponentially distributed with parameter $1/M$. In addition, the cumulative distribution function of $|h|^2 |g|^2$ is $F_{|h|^2 |g|^2}(z) = 1 - 2\sqrt{z} K_1\left(2 \sqrt{z}\right)$ \cite{JDG}. Therefore, from \eqref{op} we have
$
P_o = \E_\psi\left\{1 - 2\sqrt{\frac{\xi}{\psi}} K_1\left(2 \sqrt{\frac{\xi}{\psi}}\right)\right\},
$
where $\xi$ is given by \eqref{xi}. The expectation is evaluated with $f_\psi(\psi) = (1/M)e^{-\psi/M}$, and the proposition is proven.

\end{document}